\documentclass[journal,onecolumn]{IEEEtran}

\usepackage{amsmath,amsfonts,amssymb,amsthm,mathtools}
\usepackage{array,booktabs,graphicx}
\usepackage{enumitem}
\usepackage{hyperref}
\hypersetup{
  hidelinks,
  pdftitle={Three-Block Linear Programming Bounds for LRCs},
  pdfauthor={Ming-Hsuan Kang and Maosheng Xiong},
  pdfkeywords={locally recoverable codes, linear programming bounds, partition distance distributions, Krawtchouk polynomials, Delsarte method}
}

\theoremstyle{plain}
\newtheorem{theorem}{Theorem}[section]
\newtheorem{proposition}[theorem]{Proposition}
\newtheorem{lemma}[theorem]{Lemma}
\newtheorem{corollary}[theorem]{Corollary}
\theoremstyle{definition}
\newtheorem{definition}[theorem]{Definition}
\theoremstyle{remark}
\newtheorem{remark}[theorem]{Remark}

\newcommand{\C}{\mathcal C}
\newcommand{\F}{\mathbb F}
\newcommand{\Kraw}{K}
\newcommand{\1}{\mathbf 1}
\newcommand{\wt}{\operatorname{wt}}
\newcommand{\certfrac}[3]{\(#1/#2\,\text{ (#3)}\)}
\numberwithin{table}{section}

\begin{document}

\title{Linear Programming Bounds for Locally Recovery Codes II}

\author{Ming-Hsuan~Kang and Maosheng~Xiong
\thanks{The work of Ming-Hsuan Kang was supported by the National Science
and Technology Council, Taiwan, under Grant NSTC 115-2115-M-A49-010.
}}

\maketitle

\begin{abstract}
We give a polynomial-size linear programming bound for $q$-ary
all-symbol locally recoverable codes with locality parameters $(r,\delta)$,
without assuming linearity.  The key idea is to keep, for every ordered
pair of codewords and every selected recovery view, the joint Hamming
weight on the helper set, the recovered coordinate, and the rest of the
code -- rather than collapsing this triple into a single distance, as
earlier formulations do.  Averaging this three-block distribution over
recovery views of the same length yields exact identities linking it to
the global distance distribution, together with nonnegative
product-Krawtchouk constraints that encode locality and spectral
positivity simultaneously.  The resulting LP has polynomially many
variables, its optimum dominates the ordinary Delsarte bound, and an
earlier outside-distance formulation, the convex-hull bound of
Li--Wei--Xiong, and the dual-based bound of Gruica--Jany--Ravagnani all
arise from it as coarser marginals.  Exact rational certificates over
$q=2,3,4$ show the bound is strictly stronger than the best of these prior
LPs in thirteen of fifteen tested cases, pinning down seven exact maximum
code sizes and twelve exact maximum linear dimensions.
\end{abstract}

\begin{IEEEkeywords}
Locally recoverable codes, linear programming bounds, partition distance
distributions, Krawtchouk polynomials, Delsarte method.
\end{IEEEkeywords}

\section{Introduction}
\label{sec:introduction}

Locally recoverable codes (LRCs) allow each erased coordinate to be
reconstructed from a short local view.  They were introduced for distributed
storage, where locality reduces the number of surviving symbols accessed
during repair \cite{gopalan2012,prakash2012}.  Finite-length bounds have
used shortening, recovery graphs, entropy, generalized weights, matroids,
and dual-code structure
\cite{cadambe2015,tamobargfrolov2016,huang2015,kruglik2017}.

Delsarte's linear program bounds an ordinary code through its Hamming
distance distribution and the nonnegativity of its Krawtchouk transform
\cite{delsarte1973}.  To make this method sensitive to locality, fix a
length-$n$ code $\C\subseteq\F_q^n$ with coordinate set $[n]=\{1,\ldots,n\}$,
a recovered coordinate $i\in[n]$, and a selected helper set $S_i$, and write
\[
 T_i=\{i\}\mathbin{\dot\cup}S_i.
\]
The recovery view determines the three-part coordinate partition
\[
 [n]
 =S_i\mathbin{\dot\cup}\{i\}\mathbin{\dot\cup}O_i,
 \qquad O_i=[n]\setminus T_i.
\]
For an ordered pair of codewords $(x,y)\in\C^2$, write $d_E(x,y)$ for their
Hamming distance restricted to a coordinate set $E\subseteq[n]$, $d_H(x,y)$
for their full Hamming distance, and $\1_{\{x_i\ne y_i\}}$ for the indicator
that they differ at coordinate $i$ (all made precise in
Section~\ref{sec:preliminaries}).  We retain the three weights
\[
 \bigl(d_{S_i}(x,y),\ \1_{\{x_i\ne y_i\}},\ d_{O_i}(x,y)\bigr).
\]
They record the difference on the helper set, at the recovered coordinate,
and outside the selected local view, and their sum is $d_H(x,y)$.

At the level of an individual pair, the outside distance is determined by
the global and local distances, so nothing is lost pair by pair.  The
distinction becomes meaningful only after aggregation: the global distance
distribution and a local distance marginal each forget which global and
local values occurred together for a given pair and recovery view, while
retaining all three block weights preserves that coupling.  This is the
finest block-weight description considered here, short of tracking
individual difference vectors.

The coordinate-partition form of the MacWilliams--Delsarte transform is
classical \cite{simonis1995}.  We apply it to the family of partitions
selected by the recovery views.  After averaging views having the same
length, the resulting LP has three ingredients.  Exact centered identities
connect every global distance layer to the three-block distribution.
Local-distance zeros and a projection-collision inequality impose the
recovery requirement.  Product-Krawtchouk positivity supplies the spectral
constraints.  The number of variables is polynomial in $n$ and
$r+\delta$.

The contributions are as follows.

\begin{enumerate}[label=(\roman*)]
\item We formulate the three-block LP for all-symbol $(r,\delta)$-locality
over every finite field, without assuming linearity, uniform recovery-view
length, or disjointness among recovery views.

\item We show that the global Delsarte transform, and its two
coordinate-centered sectors, decompose exactly into nonnegative
three-block components; aggregating these recovers the outside-distance
marginal of an earlier formulation, the convex-hull variables of
Li--Wei--Xiong (LWX), and the dual statistic of Gruica--Jany--Ravagnani
(GJR) as coarser projections.

\item We give exact rational bounds for fifteen parameter sets over
$q=2,3,4$.  The three-block LP is strictly stronger than the
outside-distance marginal in thirteen cases.  Verified constructions
determine seven exact maximum code sizes and twelve exact maximum linear
dimensions.
\end{enumerate}

Gruica, Jany, and Ravagnani's bound is specific to linear codes
\cite{gruica2026}, while Li, Wei, and Xiong retain factorial moments
through order $\delta-2$ for arbitrary, possibly nonlinear, LRCs
\cite{lwx2026}.  Agarwal et al.\ obtain alphabet-dependent bounds under a
disjoint repair-group assumption \cite{agarwal2018}; no such assumption is
made here.

Definition~\ref{def:three-block-lp} in Section~\ref{sec:three-block} gives
a complete, implementation-ready statement of the LP for readers who wish
to skip the proofs.

Section~\ref{sec:preliminaries} fixes the locality convention and proves the
coordinate-partition positivity used below.  Section~\ref{sec:three-block}
constructs the three-block distribution and gives the complete LP.
Section~\ref{sec:projections} develops its spectral decomposition and
selected projections.  Section~\ref{sec:results} presents the certified
bounds and attaining constructions.

\section{Locality and coordinate-partition positivity}
\label{sec:preliminaries}

Let $q$ be a prime power and let $\C\subseteq\F_q^n$ be a nonempty code of
size $M$; linearity is not assumed.  For a code $\C'\subseteq\F_q^n$, write
$d(\C')$ for its minimum Hamming distance, with the convention $d(\C')=0$ if
$\C'$ has fewer than two distinct words; we assume $\C$ has minimum distance
$d(\C)$ at least $d$ for a fixed target $d$.  Write
$[n]=\{1,\ldots,n\}$.  For $z\in\F_q^n$, its Hamming weight is
\[
 \wt(z)=|\{u\in[n]:z_u\ne0\}|.
\]
For $E\subseteq[n]$, let $\C|_E$ denote coordinate projection onto $E$ and
put
\[
 d_E(x,y)=\wt\bigl((x-y)|_E\bigr).
\]
In particular, $d_H=d_{[n]}$.

\begin{definition}[all-symbol local-distance locality]
\label{def:locality}
The code $\C$ has all-symbol $(r,\delta)$-locality if, for every $i\in[n]$,
there is a selected helper set $S_i\subseteq[n]\setminus\{i\}$ such that,
with
\[
 T_i=\{i\}\mathbin{\dot\cup}S_i,
\]
we have
\[
 |T_i|\le r+\delta-1,
 \qquad
 d(\C|_{T_i})\ge\delta.
\]
The set $T_i$ is the selected recovery view for coordinate $i$.
\end{definition}

Informally, $\delta$ is the number of erasures the recovery view can
tolerate on its own (its local minimum distance), while $r$ controls how
large that view may be: when $\delta=2$ this is the classical notion of
locality with $r$ helper symbols, and increasing $\delta$ allows a larger
view, of size up to $r+\delta-1$, in exchange for tolerating more
simultaneous local erasures.

Throughout, we assume
\[
 r\ge1,
 \qquad
 d\ge\delta\ge2,
\]
and select one recovery view $T_i$ for every coordinate.  Put
\[
 t_i=|T_i|,
 \qquad
 R_0=\min\{n,r+\delta-1\}.
\]
Then $\delta\le t_i\le R_0$.  The selected views may have different
lengths and may overlap arbitrarily.

The assumption $d\ge\delta$ entails no loss of generality.  Indeed,
locality already implies global minimum distance at least $\delta$: if
$x_i\ne y_i$, then $x|_{T_i}\ne y|_{T_i}$ and hence
\[
 d_H(x,y)\ge d_{T_i}(x,y)\ge\delta.
\]

The normalized global distance distribution is
\[
 A_\ell
 =\frac1M|\{(x,y)\in\C^2:d_H(x,y)=\ell\}|,
 \qquad 0\le\ell\le n.
\]
Thus
\[
 A_0=1,
 \qquad
 A_\ell=0\quad(1\le\ell<d),
 \qquad
 \sum_{\ell=0}^nA_\ell=M.
\]

For the $q$-ary Krawtchouk polynomial, we use
\begin{equation}
 \Kraw_w(j;m,q)
 =\sum_{u=0}^{w}(-1)^u(q-1)^{w-u}
  \binom{j}{u}\binom{m-j}{w-u}.
\label{eq:krawtchouk}
\end{equation}
A binomial coefficient is understood to be zero when its lower index is
outside its natural range.

For each selected recovery view, put $O_i=[n]\setminus T_i$.  The LP below
uses Lemma~\ref{lem:partition-positive} with the three blocks
\[
 (P_1,P_2,P_3)=(S_i,O_i,\{i\}),
\]
whose lengths are $t_i-1$, $n-t_i$, and $1$.

\begin{remark}[why keep $i$ separate from $S_i$]
Separating $i$ from $S_i$ adds no information beyond the complete
difference vector on $T_i$: it only prevents Hamming-weight compression
from forgetting whether the recovered coordinate itself participates in
the difference, a distinction the locality constraints in
Section~\ref{sec:three-block} depend on.
\end{remark}

\begin{lemma}[coordinate-partition positivity]
\label{lem:partition-positive}
Let
\[
 P_1\mathbin{\dot\cup}\cdots\mathbin{\dot\cup}P_h=[n],
 \qquad |P_v|=m_v,
\]
and define the partition distance distribution by
\[
 D_{j_1,\ldots,j_h}
 =\frac1M
 \left|
 \left\{(x,y)\in\C^2:
 d_{P_v}(x,y)=j_v\text{ for every }v\right\}
 \right|.
\]
For integers $0\le w_v\le m_v$, let
\[
 \Omega_{\mathbf w}
 =\left\{\alpha\in\F_q^n:
 \wt(\alpha|_{P_v})=w_v\text{ for every }v\right\},
 \qquad \mathbf w=(w_1,\ldots,w_h).
\]
Fix a nontrivial additive character $\chi$ of $\F_q$ and, for
$\alpha\in\F_q^n$, write
\[
 \chi_\alpha(x)=\chi\left(\sum_{u=1}^n\alpha_u x_u\right).
\]
Then
\begin{equation}
\begin{split}
 &\sum_{j_1,\ldots,j_h}D_{j_1,\ldots,j_h}
   \prod_{v=1}^h\Kraw_{w_v}(j_v;m_v,q)\\
 &\hspace{25mm}=
 \sum_{\alpha\in\Omega_{\mathbf w}}
 \frac1M\left|\sum_{x\in\C}\chi_\alpha(x)\right|^2
 \ge0.
\end{split}
\label{eq:partition-fourier-identity}
\end{equation}
\end{lemma}

\begin{proof}
We prove the equality; nonnegativity then follows because the right-hand
side is a sum of squared absolute values.  For fixed
$\alpha\in\F_q^n$,
\begin{align*}
 \frac1M\left|\sum_{x\in\C}\chi_\alpha(x)\right|^2
 &=\frac1M\sum_{x,y\in\C}
   \chi_\alpha(x)\overline{\chi_\alpha(y)}\\
 &=\frac1M\sum_{x,y\in\C}\chi_\alpha(x-y),
\end{align*}
because $\overline{\chi_\alpha(y)}=\chi_\alpha(-y)$.  Summing over
$\Omega_{\mathbf w}$ and interchanging finite sums gives
\begin{equation}
 \sum_{\alpha\in\Omega_{\mathbf w}}
 \frac1M\left|\sum_{x\in\C}\chi_\alpha(x)\right|^2
 =\frac1M\sum_{x,y\in\C}
  \sum_{\alpha\in\Omega_{\mathbf w}}\chi_\alpha(x-y).
\label{eq:expanded-energy}
\end{equation}

Fix $(x,y)$ and write $z=x-y$.  Since the blocks are disjoint, choosing
$\alpha\in\Omega_{\mathbf w}$ is equivalent to choosing independently
$\beta_v\in\F_q^{P_v}$ with $\wt(\beta_v)=w_v$.  Hence
\begin{equation}
 \sum_{\alpha\in\Omega_{\mathbf w}}\chi_\alpha(z)
 =\prod_{v=1}^h
 \left(
 \sum_{\substack{\beta_v\in\F_q^{P_v}\\\wt(\beta_v)=w_v}}
 \chi\left(\sum_{u\in P_v}\beta_{v,u}z_u\right)
 \right).
\label{eq:block-factorization}
\end{equation}

We evaluate one factor.  Put
$j_v=\wt(z|_{P_v})=d_{P_v}(x,y)$.  The elementary identity
\begin{equation}
 \sum_{\lambda\in\F_q^\times}\chi(\lambda\tau)
 =\begin{cases}
 q-1,&\tau=0,\\
 -1,&\tau\ne0
 \end{cases}
\label{eq:nonzero-character-sum}
\end{equation}
follows because multiplication by nonzero $\tau$ permutes $\F_q^\times$ and
$\sum_{\mu\in\F_q}\chi(\mu)=0$ for nontrivial $\chi$.

Suppose exactly $\sigma$ coordinates of the support of $\beta_v$ lie among
the $j_v$ nonzero coordinates of $z|_{P_v}$.  The support can be chosen
in
\[
 \binom{j_v}{\sigma}\binom{m_v-j_v}{w_v-\sigma}
\]
ways.  By \eqref{eq:nonzero-character-sum}, these coordinates contribute
$(-1)^\sigma$, while the other $w_v-\sigma$ support coordinates contribute
$(q-1)^{w_v-\sigma}$.  Therefore
\begin{align*}
 &\sum_{\substack{\beta_v\in\F_q^{P_v}\\\wt(\beta_v)=w_v}}
 \chi\left(\sum_{u\in P_v}\beta_{v,u}z_u\right)\\
 &\qquad=\sum_{\sigma=0}^{w_v}(-1)^\sigma(q-1)^{w_v-\sigma}
 \binom{j_v}{\sigma}\binom{m_v-j_v}{w_v-\sigma}
 =\Kraw_{w_v}(j_v;m_v,q).
\end{align*}
Substituting into \eqref{eq:block-factorization} gives
\[
 \sum_{\alpha\in\Omega_{\mathbf w}}\chi_\alpha(x-y)
 =\prod_{v=1}^h
 \Kraw_{w_v}\bigl(d_{P_v}(x,y);m_v,q\bigr).
\]
Insert this into \eqref{eq:expanded-energy} and group ordered pairs by
the tuple $(d_{P_1}(x,y),\ldots,d_{P_h}(x,y))$.  The result is the
left-hand side of \eqref{eq:partition-fourier-identity}.
\end{proof}

If $\C$ is linear, the right-hand side of
\eqref{eq:partition-fourier-identity} is $|\C|$ times the corresponding
partitioned dual weight distribution.  The character proof is used because
it applies equally to nonlinear codes.

\section{The three-block distance distribution and LP}
\label{sec:three-block}

For each coordinate $i$, the selected recovery view determines
\[
 [n]=S_i\mathbin{\dot\cup}O_i\mathbin{\dot\cup}\{i\},
 \qquad O_i=[n]\setminus T_i.
\]
For $(x,y)\in\C^2$, put
\[
 a=d_{S_i}(x,y),
 \qquad
 c=d_{O_i}(x,y),
 \qquad
 b=\1_{\{x_i\ne y_i\}}.
\]
Then $d_H(x,y)=a+c+b$.  For
\[
 0\le a<t_i,\qquad 0\le c\le n-t_i,\qquad b\in\{0,1\},
\]
define
\[
 J^{(i)}_{a,c,b}
 =\frac1M
 \left|
 \left\{(x,y)\in\C^2:
 d_{S_i}(x,y)=a,\ d_{O_i}(x,y)=c,\
 \1_{\{x_i\ne y_i\}}=b\right\}
 \right|.
\]
We average selected views according to their local length:
\begin{equation}
 Y_{t,a,c,b}
 =\frac1n\sum_{\substack{i\in[n]\\t_i=t}}J^{(i)}_{a,c,b},
\label{eq:y-definition}
\end{equation}
where
\[
 \delta\le t\le R_0,
 \qquad 0\le a<t,
 \qquad 0\le c\le n-t,
 \qquad b\in\{0,1\}.
\]
We call $(Y_{t,a,c,b})$ the \emph{three-block distance distribution}.

\begin{proposition}[centered identities]
\label{prop:centered}
For every $0\le\ell\le n$,
\begin{align}
 \sum_{\substack{t,a,c\\a+c=\ell}}Y_{t,a,c,0}
 &=\frac{n-\ell}{n}A_\ell,
\label{eq:center-same}\\
 \sum_{\substack{t,a,c\\a+c+1=\ell}}Y_{t,a,c,1}
 &=\frac{\ell}{n}A_\ell.
\label{eq:center-different}
\end{align}
\end{proposition}

\begin{proof}
Fix an ordered pair at distance $\ell$.  Among the $n$ choices of the
distinguished coordinate, the pair agrees at exactly $n-\ell$ coordinates
and differs at exactly $\ell$ coordinates.  For each choice, the remaining
difference is divided between $S_i$ and $O_i$.  Sum over all ordered pairs
and use \eqref{eq:y-definition}.
\end{proof}

In particular,
\begin{equation}
 \sum_{t,a,c,b}Y_{t,a,c,b}=M,
 \qquad
 \sum_t Y_{t,0,0,0}=1.
\label{eq:three-block-mass}
\end{equation}

Since the distance inside $T_i$ is $a+b$, locality gives
\begin{equation}
 Y_{t,a,c,b}=0
 \qquad\text{whenever }1\le a+b<\delta.
\label{eq:local-zero}
\end{equation}

\begin{proposition}[projection-collision inequality]
\label{prop:collision}
For every $\delta\le t\le R_0$,
\begin{equation}
 q^{\delta-1}\sum_{a,c,b}Y_{t,a,c,b}
 \le q^t\sum_c Y_{t,0,c,0}.
\label{eq:collision}
\end{equation}
\end{proposition}

\begin{proof}
Fix $i$ with $t_i=t$ and partition $\C$ into the fibres of the projection
$\C\to\C|_{T_i}$ (the preimage classes of its distinct projected values).
If their sizes are $f_1,\ldots,f_N$, then
\[
 \sum_cJ^{(i)}_{0,c,0}=\frac1M\sum_{u=1}^Nf_u^2.
\]
The projected code has length $t$ and minimum distance at least $\delta$,
so the Singleton bound gives $N\le q^{t-\delta+1}$.  By
Cauchy--Schwarz,
\[
 \sum_{u=1}^Nf_u^2\ge\frac{M^2}{N}
 \ge M^2q^{-(t-\delta+1)}.
\]
Hence
\[
 q^t\sum_cJ^{(i)}_{0,c,0}\ge q^{\delta-1}M.
\]
Sum over all $i$ with $t_i=t$, divide by $n$, and use
$\sum_{a,c,b}J^{(i)}_{a,c,b}=M$.
\end{proof}

For
\[
 0\le p<t,
 \qquad 0\le s\le n-t,
 \qquad e\in\{0,1\},
\]
define the product-Krawtchouk transform
\begin{equation}
 F_{t,p,s,e}
 =\sum_{a,c,b}Y_{t,a,c,b}
 \Kraw_p(a;t-1,q)\Kraw_s(c;n-t,q)\Kraw_e(b;1,q).
\label{eq:F-definition}
\end{equation}

\begin{proposition}[three-block positivity]
\label{prop:three-block-positive}
For every index in the displayed ranges,
\begin{equation}
 F_{t,p,s,e}\ge0.
\label{eq:F-positive}
\end{equation}
\end{proposition}

\begin{proof}
For each $i$ with $t_i=t$, apply
Lemma~\ref{lem:partition-positive} to
$S_i\mathbin{\dot\cup}O_i\mathbin{\dot\cup}\{i\}$ and average the
resulting inequalities over all such coordinates.
\end{proof}

\subsection{Complete LP formulation}
\label{subsec:three-block-lp}

We now collect the variables, index ranges, and constraints in a form that
can be implemented directly.  Fix
\[
 q,\qquad n,\qquad d\ge\delta\ge2,\qquad r\ge1,
 \qquad R_0=\min\{n,r+\delta-1\}.
\]
No code or explicit choice of recovery views is required to state or solve
the LP: these five parameters alone determine every variable, index range,
and constraint below.  The primary variables are
\[
 A_\ell\quad(0\le\ell\le n)
\]
and
\[
 Y_{t,a,c,b}
 \quad\left(
 \delta\le t\le R_0,\ 0\le a<t,\ 0\le c\le n-t,\ b\in\{0,1\}
 \right).
\]
For implementation, each $F_{t,p,s,e}$ is the linear form
\[
 F_{t,p,s,e}
 =\sum_{a=0}^{t-1}\sum_{c=0}^{n-t}\sum_{b=0}^{1}
 Y_{t,a,c,b}
 \Kraw_p(a;t-1,q)\Kraw_s(c;n-t,q)\Kraw_e(b;1,q).
\]

\begin{definition}[three-block LP]
\label{def:three-block-lp}
The \emph{three-block LP} maximizes
\begin{equation}
 \sum_{\ell=0}^nA_\ell
\label{eq:three-block-objective}
\end{equation}
over the primary variables subject to the following constraints.

\medskip
\noindent
\textbf{Global distance:}
\begin{align}
 A_0&=1,
\label{eq:lp-A0}\\
 A_\ell&\ge0
 &&(0\le\ell\le n),
\label{eq:lp-A-nonnegative}\\
 A_\ell&=0
 &&(1\le\ell<d).
\label{eq:lp-A-zero}
\end{align}

\medskip
\noindent
\textbf{Centered identities:} for $0\le\ell\le n$,
\begin{align}
 \sum_{\substack{\delta\le t\le R_0\\
                  0\le a<t,\ 0\le c\le n-t\\a+c=\ell}}
 Y_{t,a,c,0}
 &=\frac{n-\ell}{n}A_\ell,
\label{eq:lp-center-same}\\
 \sum_{\substack{\delta\le t\le R_0\\
                  0\le a<t,\ 0\le c\le n-t\\a+c+1=\ell}}
 Y_{t,a,c,1}
 &=\frac{\ell}{n}A_\ell.
\label{eq:lp-center-different}
\end{align}

\medskip
\noindent
\textbf{Locality:}
\begin{align}
 Y_{t,a,c,b}&\ge0,
\label{eq:lp-y-nonnegative}\\
 Y_{t,a,c,b}&=0
 &&\text{if }1\le a+b<\delta,
\label{eq:lp-local-zero}\\
 q^{\delta-1}\sum_{a=0}^{t-1}\sum_{c=0}^{n-t}\sum_{b=0}^{1}
 Y_{t,a,c,b}
 &\le q^t\sum_{c=0}^{n-t}Y_{t,0,c,0}
 &&(\delta\le t\le R_0).
\label{eq:lp-collision}
\end{align}

\medskip
\noindent
\textbf{Three-block Fourier positivity:}
\begin{equation}
 F_{t,p,s,e}\ge0
\label{eq:lp-fourier}
\end{equation}
for
\[
 \delta\le t\le R_0,
 \qquad 0\le p<t,
 \qquad 0\le s\le n-t,
 \qquad e\in\{0,1\}.
\]
The optimum is denoted by
\[
 U_{\mathrm{3blk}}(q,n,d,r,\delta).
\]
\end{definition}

The centered identities imply
\[
 \sum_{t,a,c,b}Y_{t,a,c,b}=\sum_{\ell=0}^nA_\ell,
\]
so no separate total-mass constraint is needed.

\begin{theorem}[validity of the three-block LP]
\label{thm:three-block-validity}
Every $q$-ary code with minimum distance at least $d$ and all-symbol
$(r,\delta)$-locality satisfies
\[
 |\C|\le U_{\mathrm{3blk}}(q,n,d,r,\delta).
\]
No linearity assumption is required.
\end{theorem}

\begin{proof}
Choose one admissible recovery view for each coordinate and construct
$A_\ell$ and $Y_{t,a,c,b}$ from the code.  The centered identities,
local-distance zeros, collision inequalities, and Fourier inequalities
were proved above.  The resulting feasible point has objective value
$\sum_\ell A_\ell=|\C|$.
\end{proof}

The LP has
\[
 n+1+2\sum_{t=\delta}^{R_0}t(n-t+1)=O(nR_0^2)
\]
primary variables.

\begin{remark}[linear-code interpretation]
\label{rem:linear}
If $\C$ is linear, let $A^\perp_{i;p,s,e}$ denote the number of dual
codewords having weights $p,s,e$ on $S_i,O_i,\{i\}$, respectively.  Then
\[
 F_{t,p,s,e}
 =\frac{|\C|}{n}
 \sum_{\substack{i\in[n]\\t_i=t}}A^\perp_{i;p,s,e}.
\]
This interpretation is separate from the arbitrary-code validity proof.
\end{remark}

\section{Spectral decomposition and selected projections}
\label{sec:projections}

The variables $Y_{t,a,c,b}$ can be made smaller by summing over indices or
retaining selected linear moments.  There are many such projections.  We
first derive two exact spectral identities used throughout, and then record
three projections relevant for comparison: the global-distance
projection, the outside-distance marginal used in an earlier formulation,
and the moment statistics appearing in LWX and GJR.

The Krawtchouk generating function
\begin{equation}
 \sum_{w=0}^m\Kraw_w(j;m,q)z^w
 =(1+(q-1)z)^{m-j}(1-z)^j
\label{eq:kraw-generating}
\end{equation}
gives
\begin{equation}
 \Kraw_w(a+c+b;n,q)
 =\sum_{p+s+e=w}
 \Kraw_p(a;t-1,q)\Kraw_s(c;n-t,q)\Kraw_e(b;1,q).
\label{eq:product-kraw}
\end{equation}
Define the global transform
\[
 B_w=\sum_{\ell=0}^nA_\ell\Kraw_w(\ell;n,q).
\]
The ordinary Delsarte LP maximizes $\sum_\ell A_\ell$ subject to
$A_0=1$, $A_\ell\ge0$, $A_\ell=0$ for $1\le\ell<d$, and
$B_w\ge0$ for $0\le w\le n$.  Its optimum is denoted by
$U_{\mathrm{Del}}(q,n,d)$.

\begin{theorem}[global spectral decomposition]
\label{thm:global-decomposition}
Every feasible point of the three-block LP satisfies
\begin{equation}
 B_w=\sum_t\sum_{p+s+e=w}F_{t,p,s,e}.
\label{eq:global-decomposition}
\end{equation}
In particular, $B_w\ge0$, so the ordinary Delsarte inequalities are
implied by the three-block LP.
\end{theorem}

\begin{proof}
Adding the two centered identities gives, for every function $f$,
\[
 \sum_{t,a,c,b}Y_{t,a,c,b}f(a+c+b)
 =\sum_{\ell=0}^nA_\ell f(\ell).
\]
Take $f(\ell)=\Kraw_w(\ell;n,q)$, substitute
\eqref{eq:product-kraw}, and use \eqref{eq:F-definition}.
\end{proof}

\begin{proposition}[distinguished-coordinate sectors]
\label{prop:sector-identities}
For $1\le w\le n$,
\begin{equation}
 \sum_t\sum_{p+s=w-1}F_{t,p,s,1}=\frac{w}{n}B_w,
\label{eq:nonzero-sector}
\end{equation}
and, for $0\le w<n$,
\begin{equation}
 \sum_t\sum_{p+s=w}F_{t,p,s,0}=\frac{n-w}{n}B_w.
\label{eq:zero-sector}
\end{equation}
\end{proposition}

\begin{proof}
By the two-block version of \eqref{eq:product-kraw},
\begin{align*}
 \sum_{p+s=w-1}F_{t,p,s,1}
 =\sum_{a,c,b}Y_{t,a,c,b}
 \Kraw_{w-1}(a+c;n-1,q)\Kraw_1(b;1,q).
\end{align*}
Since $\Kraw_1(0;1,q)=q-1$ and $\Kraw_1(1;1,q)=-1$, summing over $t$ and
using \eqref{eq:center-same}--\eqref{eq:center-different} gives
\begin{align*}
 &\sum_t\sum_{p+s=w-1}F_{t,p,s,1}\\
 &\quad=\frac1n\sum_{\ell=0}^nA_\ell
 \Bigl((n-\ell)(q-1)\Kraw_{w-1}(\ell;n-1,q)
 -\ell\Kraw_{w-1}(\ell-1;n-1,q)\Bigr).
\end{align*}
The standard recurrence
\[
 w\Kraw_w(\ell;n,q)
 =(n-\ell)(q-1)\Kraw_{w-1}(\ell;n-1,q)
 -\ell\Kraw_{w-1}(\ell-1;n-1,q)
\]
proves \eqref{eq:nonzero-sector}.  Similarly,
\begin{align*}
 \sum_t\sum_{p+s=w}F_{t,p,s,0}
 =\frac1n\sum_{\ell=0}^nA_\ell
 \Bigl((n-\ell)\Kraw_w(\ell;n-1,q)
 +\ell\Kraw_w(\ell-1;n-1,q)\Bigr),
\end{align*}
and the recurrence
\[
 (n-w)\Kraw_w(\ell;n,q)
 =(n-\ell)\Kraw_w(\ell;n-1,q)
 +\ell\Kraw_w(\ell-1;n-1,q)
\]
proves \eqref{eq:zero-sector}.
\end{proof}

\subsection{Global-distance and outside-distance projections}

Retaining only $\ell=a+c+b$ gives the ordinary global distance
distribution.  Theorem~\ref{thm:global-decomposition} shows that its
Delsarte positivity is an exact sum of three-block inequalities.

For comparison with the earlier formulation that motivated this work,
sum out the outside distance:
\begin{equation}
 X_{t,a,b}=\sum_{c=0}^{n-t}Y_{t,a,c,b}.
\label{eq:outside-marginal}
\end{equation}
The marginal $X_{t,a,b}$ retains the helper distance and the indicator at
the recovered coordinate.  Put
\[
 m_t=\sum_{a,b}X_{t,a,b},
\]
and define
\begin{align*}
 R_{t,p}&=\sum_{a,b}X_{t,a,b}\Kraw_p(a;t-1,q)
 =F_{t,p,0,0},\\
 \Theta_{t,w}&=\sum_{a,b}X_{t,a,b}
 \Kraw_{w-1}(a;t-1,q)\Kraw_1(b;1,q)
 =F_{t,w-1,0,1}.
\end{align*}
We set $R_{t,p}=0$ for $p\ge t$ and $\Theta_{t,w}=0$ for $w>t$.
Selecting the terms with outside degree $s=0$ in
\eqref{eq:nonzero-sector}--\eqref{eq:zero-sector} gives the exact
weight-one relation
\begin{equation}
 \sum_t\Theta_{t,1}=\frac1nB_1,
\label{eq:outside-U-charge-one}
\end{equation}
and, after discarding the nonnegative terms with $s>0$ at higher weights,
\begin{align}
 \sum_t\Theta_{t,w}&\le\frac{w}{n}B_w
 &&(2\le w\le n),
\label{eq:outside-U-charge}\\
 \sum_tR_{t,w}&\le\frac{n-w}{n}B_w
 &&(1\le w<n).
\label{eq:outside-R-charge}
\end{align}

The \emph{outside-distance marginal LP} maximizes $\sum_\ell A_\ell$
over nonnegative $A_\ell$ and $X_{t,a,b}$, subject to
\begin{gather}
 A_0=1,
 \qquad A_\ell=0\ (1\le\ell<d),
 \qquad B_w\ge0,
\label{eq:out-global}\\
 \sum_tm_t=\sum_\ell A_\ell,
 \qquad
 \sum_{t,a}X_{t,a,1}=\frac1n\sum_\ell\ell A_\ell,
\label{eq:out-centered}\\
 X_{t,a,b}=0\quad(1\le a+b<\delta),
 \qquad
 q^t X_{t,0,0}\ge q^{\delta-1}m_t,
 \qquad
 \sum_tX_{t,0,0}\ge1,
\label{eq:out-local}\\
 R_{t,p}\ge0,
 \qquad \Theta_{t,w}\ge0.
\label{eq:out-spectral}
\end{gather}
Here $0\le p<t$ and $1\le w\le t$ in
\eqref{eq:out-spectral}; the LP also includes
\eqref{eq:outside-U-charge-one}--\eqref{eq:outside-R-charge}.
Its optimum is denoted by $U_{\mathrm{out}}(q,n,d,r,\delta)$.

\begin{proposition}[outside-distance projection]
\label{prop:outside-projection}
For all admissible parameters,
\begin{equation}
 U_{\mathrm{3blk}}\le U_{\mathrm{out}}\le U_{\mathrm{Del}}.
\label{eq:outside-hierarchy}
\end{equation}
\end{proposition}

\begin{proof}
Summing the three-block variables over $c$ gives
\eqref{eq:out-centered}--\eqref{eq:out-local}.  The identities defining
$R$ and $\Theta$ give their nonnegativity, while
\eqref{eq:outside-U-charge-one} is the weight-one sector identity and
\eqref{eq:outside-U-charge}--\eqref{eq:outside-R-charge} follow by
discarding nonnegative terms with $s>0$ from the exact sector identities.
Discarding $X$ gives the ordinary Delsarte LP.
\end{proof}

The first inequality can be strict because the outside-distance marginal
discards every spectral component with $s>0$.  This projection is included
because it recovers the earlier formulation; it is only one of many
possible aggregations of the three-block variables.  The hierarchy is
extended to include the LWX and GJR comparisons in
Section~\ref{sec:prior-projections}.

\subsection{Moment summaries: LWX and GJR}
\label{sec:prior-projections}

The LWX formulation retains factorial moments of the nonzero-at-$i$
sector.  The following length-resolved identity is the three-block
projection corresponding to the nonlinear factorial-moment identity of
Li--Wei--Xiong \cite[Theorem~3]{lwx2026}.

\begin{lemma}[pointed local moments]
\label{lem:pointed-moments}
For every $0\le m\le\delta-2$,
\begin{equation}
 \sum_{w=1}^{t}\binom{w-1}{m}\Theta_{t,w}
 =q^t(1-q^{-1})^{m+1}\binom{t-1}{m}X_{t,0,0}.
\label{eq:pointed-moments}
\end{equation}
\end{lemma}

\begin{proof}
Fix $j$ and put $N=t-1$.  By the generating function
\eqref{eq:kraw-generating},
\[
 G_j(z)=\sum_{p=0}^N\Kraw_p(j;N,q)z^p
 =(1+(q-1)z)^{N-j}(1-z)^j.
\]
Differentiating $m$ times and evaluating at $z=1$ isolates the
coefficients weighted by $\binom pm$, since
$G_j^{(m)}(1)=m!\sum_p\binom pm\Kraw_p(j;N,q)$.  Write
$G_j(z)=(1-z)^jH_j(z)$ with $H_j(z)=(1+(q-1)z)^{N-j}$.  The first
$j-1$ derivatives of $(1-z)^j$ vanish at $z=1$, so by the Leibniz
rule the only surviving term in $G_j^{(m)}(1)$ is the one in which
all $j$ derivatives fall on the factor $(1-z)^j$; this requires
$j\le m$.  Evaluating that term gives
\begin{equation}
 \sum_{p=0}^N\binom pm\Kraw_p(j;N,q)
 =(-1)^j\binom{N-j}{m-j}(q-1)^{m-j}q^{N-m}
\label{eq:kraw-moment}
\end{equation}
for $0\le j\le m$, and $0$ for $j>m$.

Apply \eqref{eq:kraw-moment} with $N=t-1$ and $j=a$, and substitute
into the definition of $\Theta_{t,w}$:
\begin{align*}
 \sum_{w=1}^t\binom{w-1}{m}\Theta_{t,w}
 &=\sum_{a,b}X_{t,a,b}\Kraw_1(b;1,q)
   \sum_{p=0}^{t-1}\binom pm\Kraw_p(a;t-1,q)\\
 &=\sum_{a=0}^{m}\sum_bX_{t,a,b}\Kraw_1(b;1,q)
   (-1)^a\binom{t-1-a}{m-a}(q-1)^{m-a}q^{t-1-m},
\end{align*}
where the outer sum truncates at $a=m$ because
\eqref{eq:kraw-moment} vanishes for $a>m$.

Since $m\le\delta-2$, every index $0\le a\le m$ satisfies
$1\le a+1<\delta$, so the local-distance zeros \eqref{eq:local-zero}
give $X_{t,a,1}=0$ for every such $a$: the entire $b=1$ contribution
vanishes.  For $b=0$ the same zeros give $X_{t,a,0}=0$ whenever
$1\le a<\delta$, hence for every $1\le a\le m$; only $a=0$ survives.
The sum therefore collapses to the single term $a=0$, $b=0$:
\[
 \sum_{w=1}^t\binom{w-1}{m}\Theta_{t,w}
 =X_{t,0,0}\,\Kraw_1(0;1,q)\binom{t-1}{m}
  (q-1)^mq^{t-1-m}.
\]
Since $\Kraw_1(0;1,q)=q-1$, this equals
\[
 X_{t,0,0}\binom{t-1}{m}(q-1)^{m+1}q^{t-1-m}
 =q^t(1-q^{-1})^{m+1}\binom{t-1}{m}X_{t,0,0},
\]
as claimed.
\end{proof}

Put
\[
 v_w=\sum_t\Theta_{t,w}\quad(1\le w\le R_0),
 \qquad
 Z_t=(1-q^{-1})q^t X_{t,0,0}\quad(\delta\le t\le R_0).
\]
In the normalization used here, this is an extended formulation of the
unbalanced LWX convex-hull LP: the nonnegative $Z_t$ variables encode a
common mixture of admissible local lengths.  The LP consists of the global
Delsarte constraints, nonnegative $v_w,Z_t$, and
\begin{align}
 \sum_wv_w&=\sum_tZ_t,
\label{eq:lwx-mass}\\
 \sum_w\binom{w-1}{m}v_w
 &=(1-q^{-1})^m\sum_t\binom{t-1}{m}Z_t
 &&(1\le m\le\delta-2),
\label{eq:lwx-moments}\\
 \sum_wv_w&\ge(1-q^{-1})q^{\delta-1}\sum_\ell A_\ell,
\label{eq:lwx-local}\\
 B_1&=nv_1,
\label{eq:lwx-charge-one}\\
 B_w&\ge\frac{n}{w}v_w
 &&(2\le w\le R_0).
\label{eq:lwx-charge}
\end{align}
Let $U_{\mathrm{LWX}}(q,n,d,r,\delta)$ denote its optimum.

\begin{proposition}[LWX moment projection]
\label{prop:lwx-projection}
For all admissible parameters,
\begin{equation}
 U_{\mathrm{3blk}}\le U_{\mathrm{out}}
 \le U_{\mathrm{LWX}}\le U_{\mathrm{Del}}.
\label{eq:lwx-hierarchy}
\end{equation}
\end{proposition}

\begin{proof}
Sum \eqref{eq:pointed-moments} over $t$ to obtain
\eqref{eq:lwx-mass}--\eqref{eq:lwx-moments}.  The collision inequalities
imply \eqref{eq:lwx-local}, while \eqref{eq:outside-U-charge-one} and
\eqref{eq:outside-U-charge} give \eqref{eq:lwx-charge-one} and
\eqref{eq:lwx-charge}, respectively.  The remaining constraints are
Delsarte's.
\end{proof}

For a linear code, the three-block coefficients have the dual
interpretation of Remark~\ref{rem:linear}.  Before averaging over $i$,
the sum
\[
 G_{i,w}=\sum_{p+s=w-1}A^\perp_{i;p,s,1}
\]
records the total dual weight $w$ and whether the dual word is nonzero at
coordinate $i$, but not how its remaining support is divided between
$S_i$ and $O_i$.  This is the coordinate-pointed dual marginal used in
the refined-weight framework of Gruica, Jany, and Ravagnani
\cite{gruica2026}, whose coordinate-symmetrization and nondegeneracy
conditions we do not reproduce here; applying their construction to
$G_{i,w}$ yields the locality statistic behind their LP.  The comparison
in Section~\ref{sec:results} uses that specialization directly, and we
write $k_{\mathrm{GJR}}$ for the resulting dimension cap in
Table~\ref{tab:linear} (see \cite{gruica2026} for the precise
definitions); unlike $U_{\mathrm{out}}$ and $U_{\mathrm{LWX}}$ above, this
is the GJR LP applied directly rather than a size-bound projection
derived within the present LP, so no $U_{\mathrm{GJR}}$ column appears in
Table~\ref{tab:three-block-results}.  Its exact certificates are included
in the companion reproducibility repository.

Li--Wei--Xiong prove that the balanced base LP obtained by imposing
$v_1=0$ and omitting the positive-order moment constraints gives the same
dimension bound as the symmetrized GJR LP for nondegenerate linear codes
\cite[Proposition~2]{lwx2026}.  This comparison concerns the balanced base
LP, not the full higher-moment convex-hull strengthening used in
$U_{\mathrm{LWX}}$.

The outside-distance marginal, the LWX moments, and the GJR dual statistic
are selected examples rather than an exhaustive hierarchy.  Other indices
or moments of $Y_{t,a,c,b}$ may also be summed out; each aggregation reduces
computational cost while potentially discarding nonnegative spectral
information.

\section{Certified finite-length bounds}
\label{sec:results}

Let
\[
 U_{\mathrm{Del}},\qquad
 U_{\mathrm{LWX}},\qquad
 U_{\mathrm{out}},\qquad
 U_{\mathrm{3blk}}
\]
denote, respectively, the ordinary Delsarte bound, the unbalanced LWX
convex-hull bound, the outside-distance marginal bound, and the
three-block bound.  All four bound arbitrary, possibly nonlinear, codes.

Table~\ref{tab:three-block-results} gives exact LP optima for fifteen
parameter sets.  The column $M_0$ is the size of a verified linear code
with the indicated parameters.  Every displayed LP value is supported by
a rational primal--dual certificate checked independently of the numerical
optimizer.  Whenever $U_{\mathrm{3blk}}=M_0$, the maximum code size is
therefore determined without assuming linearity.

\begin{table}[!t]
\caption{Certified LP bounds and verified constructions.  Nonintegral
entries are followed by decimal approximations.  Bold entries mark rows in
which the three-block bound improves the induced linear-dimension cap
$k_X=\lfloor\log_qU_X\rfloor$ (defined in the text following the table).}
\label{tab:three-block-results}
\centering
\footnotesize
\setlength{\tabcolsep}{2.2pt}
\renewcommand{\arraystretch}{1.10}
\begin{tabular}{@{}cccccc@{}}
\toprule
$(q,n,d,r,\delta)$
& $U_{\mathrm{Del}}$
& $U_{\mathrm{LWX}}$
& $U_{\mathrm{out}}$
& $U_{\mathrm{3blk}}$
& $M_0$\\
\midrule
$(2,3,2,1,2)$
& $4$
& \certfrac{8}{3}{2.67}
& \certfrac{8}{3}{2.67}
& $2$
& $2$\\
$(2,9,3,1,2)$
& \certfrac{128}{3}{42.67}
& \certfrac{512}{33}{15.52}
& \certfrac{512}{33}{15.52}
& \certfrac{64}{5}{12.80}
& $8$\\
$(2,9,4,4,3)$
& \certfrac{128}{5}{25.60}
& \certfrac{256}{13}{19.69}
& $16$
& $\mathbf{8}$
& $8$\\
$(2,10,5,1,3)$
& $12$
& $6$
& $6$
& $6$
& $4$\\
$(2,13,7,3,4)$
& $8$
& $8$
& \certfrac{384}{49}{7.84}
& \certfrac{4304}{569}{7.56}
& $4$\\
$(2,14,4,2,2)$
& $512$
& \certfrac{11008}{55}{200.15}
& \certfrac{11008}{55}{200.15}
& \certfrac{1359968}{8555}{158.97}
& $64$\\
$(2,15,3,7,3)$
& $2048$
& $1024$
& \certfrac{19456}{41}{474.54}
& \certfrac{2048}{5}{409.60}
& $256$\\
$(2,15,5,1,3)$
& $256$
& \certfrac{3168}{149}{21.26}
& \certfrac{3168}{149}{21.26}
& \certfrac{96}{5}{19.20}
& $8$\\
$(2,20,8,3,4)$
& \certfrac{10752}{37}{290.59}
& \certfrac{1824}{29}{62.90}
& \certfrac{728}{23}{31.65}
& \certfrac{336}{11}{30.55}
& $16$\\
$(3,3,2,1,2)$
& $9$
& \certfrac{9}{2}{4.50}
& \certfrac{9}{2}{4.50}
& $3$
& $3$\\
$(3,5,2,2,2)$
& $81$
& \certfrac{243}{7}{34.71}
& \certfrac{243}{7}{34.71}
& $27$
& $27$\\
$(3,7,3,1,3)$
& \certfrac{729}{5}{145.80}
& \certfrac{243}{17}{14.29}
& \certfrac{243}{17}{14.29}
& \certfrac{81}{7}{11.57}
& $3$\\
$(3,8,3,2,3)$
& \certfrac{1701}{5}{340.20}
& $81$
& $81$
& $81$
& $81$\\
$(3,8,5,2,3)$
& \certfrac{459}{11}{41.73}
& \certfrac{315}{11}{28.64}
& \certfrac{315}{11}{28.64}
& $\mathbf{9}$
& $9$\\
$(4,5,2,2,2)$
& $256$
& \certfrac{256}{3}{85.33}
& \certfrac{256}{3}{85.33}
& $64$
& $64$\\
\bottomrule
\end{tabular}
\end{table}

The outside-distance marginal improves the LWX bound in four rows.  The
three-block bound is strictly smaller than the outside-distance marginal
in thirteen rows and equal in the remaining two.  Two strict improvements
cross a power of the alphabet:
\[
 U_{\mathrm{3blk}}(2,9,4,4,3)=8<16=U_{\mathrm{out}}(2,9,4,4,3),
\]
and
\[
 U_{\mathrm{3blk}}(3,8,5,2,3)=9
 <\frac{315}{11}=U_{\mathrm{out}}(3,8,5,2,3).
\]
Both values are attained by verified constructions.

For a linear code of dimension $k$, $|\C|=q^k$.  Define
\[
 k_X=\left\lfloor\log_qU_X\right\rfloor
\]
for each size bound $U_X$, and let $k_0=\log_qM_0$.  Let $k_{\mathrm{GJR}}$ denote the dimension cap obtained from the GJR LP.
Table~\ref{tab:linear} compares it with the resulting caps on six
representative rows.

\begin{table}[!t]
\caption{Certified upper bounds on the dimension of linear codes.}
\label{tab:linear}
\centering
\small
\setlength{\tabcolsep}{4pt}
\begin{tabular}{@{}cccccc@{}}
\toprule
$(q,n,d,r,\delta)$
& $k_{\mathrm{GJR}}$
& $k_{\mathrm{LWX}}$
& $k_{\mathrm{out}}$
& $k_{\mathrm{3blk}}$
& $k_0$\\
\midrule
$(2,9,4,4,3)$  & $4$ & $4$ & $4$ & $\mathbf{3}$ & $3$\\
$(2,10,5,1,3)$ & $3$ & $2$ & $2$ & $2$ & $2$\\
$(2,13,7,3,4)$ & $3$ & $3$ & $2$ & $2$ & $2$\\
$(2,15,3,7,3)$ & $11$ & $10$ & $8$ & $8$ & $8$\\
$(2,20,8,3,4)$ & $8$ & $5$ & $4$ & $4$ & $4$\\
$(3,8,5,2,3)$  & $3$ & $3$ & $3$ & $\mathbf{2}$ & $2$\\
\bottomrule
\end{tabular}
\end{table}

Let $M_{\max}(q,n,d,r,\delta)$ denote the largest size of a $q$-ary
length-$n$ code with minimum distance at least $d$ and all-symbol
$(r,\delta)$-locality.

\begin{theorem}[exact maximum code sizes]
\label{thm:exact-sizes}
The following maximum code sizes are exact:
\[
\begin{aligned}
 M_{\max}(2,3,2,1,2)&=2,\\
 M_{\max}(2,9,4,4,3)&=8,\\
 M_{\max}(3,3,2,1,2)&=3,\\
 M_{\max}(3,5,2,2,2)&=27,\\
 M_{\max}(3,8,3,2,3)&=81,\\
 M_{\max}(3,8,5,2,3)&=9,\\
 M_{\max}(4,5,2,2,2)&=64.
\end{aligned}
\]
\end{theorem}

\begin{proof}
The upper bounds are the corresponding three-block LP values in
Table~\ref{tab:three-block-results}.  We describe attaining linear codes.

The binary and ternary length-three cases are attained by repetition
codes, using any two-coordinate repetition puncture as the selected view.
For $(2,9,4,4,3)$, take the binary code generated by
\[
 \begin{pmatrix}
 1&0&0&1&1&0&1&1&0\\
 0&1&0&1&0&1&1&0&1\\
 0&0&1&0&1&1&1&0&0
 \end{pmatrix}.
\]
It is obtained from the $[7,3,4]$ simplex code by repeating its first two
columns.  The selected views
\[
 \{1,2,4,8,9\},\qquad
 \{1,2,3,4,5,6\},\qquad
 \{1,2,3,4,5,7\}
\]
have projected minimum distance three; use the first view for coordinates
$1,2,4,8,9$, the second for $3,5,6$, and the third for $7$.

For $q=3,4$, the length-five construction is the direct sum of a
$q$-ary $[3,2,2]$ single-parity-check code and a length-two repetition
code.  The two component supports are the selected views.  The ternary
parameters $(3,8,3,2,3)$ are attained by the direct sum of two $[4,2,3]$
MDS codes, again using the component supports as views.  Finally, repeat
every coordinate of a ternary $[4,2,3]$ MDS code twice.  The resulting
length-eight, dimension-two code has distance six.  For each coordinate,
choose that coordinate together with one copy of each of the other three
original coordinates; the resulting four-coordinate projection is the
original $[4,2,3]$ code.  This attains $(3,8,5,2,3)$.
\end{proof}

\begin{corollary}[exact maximum linear dimensions]
\label{cor:exact-linear-dimensions}
For every row of Table~\ref{tab:three-block-results} except
\[
 (2,14,4,2,2),\qquad
 (2,15,5,1,3),\qquad
 (3,7,3,1,3),
\]
the maximum dimension of a linear code is
\[
 k_{\max}=k_0=\log_qM_0.
\]
Thus the exact maximum linear dimension is determined in twelve of the
fifteen displayed parameter sets.
\end{corollary}

\begin{proof}
For each stated row, the verified construction has size $q^{k_0}=M_0$,
while the three-block bound satisfies
\[
 U_{\mathrm{3blk}}<q^{k_0+1}.
\]
Hence every linear code has dimension at most $k_0$, and the construction
shows that dimension $k_0$ is attained.
\end{proof}

For the three excluded rows, a gap remains between the dimension cap from
the three-block LP and the dimension of the verified construction.

\section{Conclusion}
\label{sec:conclusion}

A selected recovery view partitions the coordinates into the helper set,
the recovered coordinate, and the complement of the view.  Retaining the
three corresponding weights for the same ordered pair gives exact centered
identities and nonnegative product-Krawtchouk transforms.  Together with
the local-distance and projection-collision constraints, these relations
yield a polynomial-size LP for arbitrary $q$-ary, possibly nonlinear,
locally recoverable codes.  The certified computations show that preserving
the three-block coupling can strictly improve finite-length bounds over
binary, ternary, and quaternary alphabets.

The LP remains an outer relaxation of realizable recovery-view data.  It
averages views having the same local length and therefore does not retain
the identity of the recovered coordinate, exact coordinate supports, or
intersections among recovery views.  Prescribed recovery-set geometry and
several recovery views for the same coordinate are natural directions for
finer models.

\section*{Acknowledgment}

\section*{Data and Code Availability}
Exact rational certificates, optimizer-independent checkers, and
construction verifiers for Table~\ref{tab:three-block-results} accompany
the manuscript in a separate reproducibility repository.
They are available at
\begin{center}
\small\url{https://github.com/kmsming-prog/three-block-lrc-lp-reproducibility}
\end{center}

\bibliographystyle{IEEEtran}
\bibliography{references}

\end{document}